\documentclass{llncs}
\usepackage{amssymb,graphicx}
\usepackage{cite}

\spnewtheorem{Claim}{Claim}{\bfseries}{\it}
\newtheorem{observation}[theorem]{Observation}

\title{Recognizing $k$-equistable graphs in FPT time\thanks{This research is supported in part by ``Agencija za raziskovalno dejavnost Republike Slovenije'', research program P$1$--$0285$ and research projects J$1$-$5433$, J$1$-$6720$, J$1$-$6743$, and BI-FR/$15$--$16$--PROTEUS--$003$.}}

\author{Eun Jung Kim\inst{1} \and Martin Milani\v{c}\inst{2} \and Oliver Schaudt\inst{3}}
\institute{CNRS-Universit\'e Paris-Dauphine\\ 
Place du Mar\'echal de Lattre de Tassigny, 75775 Paris Cedex 16\\ \email{eun-jung.kim@dauphine.fr}  \and
University of Primorska, UP IAM and UP FAMNIT\\
Muzejski trg 2, SI-6000 Koper, Slovenia.\\ \email{martin.milanic@upr.si}
\and
Universit\"at zu K\"oln,
Institut f\"ur Informatik,
Weyertal 80, 50931 K\"oln, Germany.\\ \email{schaudto@uni-koeln.de}}

\begin{document}

\maketitle

\begin{abstract}
A graph $G = (V,E)$ is called \emph{equistable} if there exist a
positive integer $t$ and a weight function $w : V \to \mathbb{N}$ such that $S \subseteq V$ is a
maximal stable set of $G$ if and only if $w(S) = t$. Such a function $w$
is called an \emph{equistable function} of $G$. For a positive integer $k$, a graph $G = (V,E)$ is said to
be \emph{$k$-equistable} if it admits an equistable function which is bounded by $k$.

We prove that the problem of recognizing $k$-equistable graphs is fixed parameter tractable when parameterized by $k$, affirmatively answering a question of Levit et al. In fact, the problem admits an $O(k^5)$-vertex kernel that can be computed in linear time.\\

\end{abstract}

\section{Introduction}

The main notion studied in this paper is the class of equistable graphs, introduced by Payan in 1980~\cite{MR553649} as a generalization of the well known and well studied class of threshold graphs~\cite{MR0479384,MR1417258}. A {\it stable} (or {\it independent}) {\it set} in a (\hbox{finite}, simple, undirected) graph $G$ is a set of pairwise non-adjacent vertices. A {\it maximal stable set} is a stable set not contained in any other stable set. A graph $G=(V,E)$ is said to be {\it equistable} if there exists a function $\varphi:V \to \mathbb{R}_+$ such that for every $S\subseteq V$, set $S$ is a maximal stable set of $G$ if and only if $\varphi(S):= \sum_{x\in S}\varphi(x)=1$. Equivalently, $G$ is equistable if and only if there exist a positive integer $t$ and a weight function $w:V \to \mathbb{N}:= \{1,2,3,\ldots\}$ such that $S\subseteq V$ is a maximal stable set of $G$ if and only if $w(S)=t$. Such a function $w$ is called an {\em equistable function} of $G$, while the pair $(w,t)$ is called an {\em equistable structure}. Equistable graphs were studied in a series of papers~\cite{MR3040147,MR2794315,MR2379078,MR3162288,MR2024271,MR2024275,MR2823204,MR,MR1268776,MR553649,MT}; 
besides threshold graphs and cographs, they also generalize the class of general partition graphs~\cite{McAvaney1993131,MR2080087,MR2794315}. The complexity status of recognizing equistable graphs is open, and no combinatorial characterization of equistable graphs is known.

Levit et al.~introduced in~\cite{MR3040147} the notion of $k$-equistable graphs. For a positive integer $k$, a graph $G=(V,E)$ is said to be {\em $k$-equistable} if it admits an equistable function $w: V \to [k]:=\{1,\ldots, k\}$. Such a weight function is called a {\em $k$-equistable function}, and the corresponding structure $(w,t)$ is a {\em $k$-equistable structure}. We remark that there exist equistable graphs such that the smallest $k$ for which the graph is $k$-equistable is exponential in the number of vertices of $G$~\cite{MR2823204}.

For a positive integer $t$, an equistable graph $G=(V,E)$ is said to be {\em \hbox{target-$t$} equistable} if it admits an equistable
function $w: V \to \mathbb{N}$ with equistable structure $(w,t)$.
Clearly, every target-$t$ equistable graph is also $t$-equistable (but not vice versa).

\begin{sloppypar}
As mentioned above, the complexity of recognizing equistable graphs is open, but it seems \hbox{plausible} that the problem could be NP-hard. It thus makes sense to search ways to simplify the recognition problem. To this end, we consider the following two parameterized problems related to equistability.
\end{sloppypar}

\medskip
\begin{center}
\fbox{\parbox{0.82\linewidth}{\noindent
{\sc $k$-Equistability}\\[.8ex]
\begin{tabular*}{.9\textwidth}{rl}
{\em Input:} & A graph $G=(V,E)$, a positive integer $k$.\\
{\em Parameter:} & $k$.\\
{\em Question:} & Is $G$ $k$-equistable?
\end{tabular*}
}}
\end{center}

\smallskip

\begin{center}
\fbox{\parbox{0.82\linewidth}{\noindent
{\sc Target-$t$ Equistability}\\[.8ex]
\begin{tabular*}{.9\textwidth}{rl}
{\em Input:} & A graph $G=(V,E)$, a positive integer $t$.\\
{\em Parameter:} & $t$.\\
{\em Question:} & Is $G$ target-$t$ equistable?
\end{tabular*}
}}
\end{center}

Apart from being natural parameterizations of the equistability problem, the
first problem has been tackled before (in a non-parameterized variant)
in a paper by Levit et al.~\cite{MR3040147}. There they prove the following.

\begin{theorem}[Levit et al.~\cite{MR3040147}]\label{thm:k-equistable}
For every fixed $k$, there is an $O\left(n^{2k}\right)$ algorithm
to decide whether a given $n$-vertex graph is $k$-equistable.
In case of a positive instance, the algorithm also produces a $k$-equistable structure of~$G$.
\end{theorem}

Also, the authors ask whether Theorem~\ref{thm:k-equistable} can be strengthened in the sense that there is an FPT-algorithm for recognizing $k$-equistable graphs.
We answer this question affirmatively.

More precisely, we prove the following results:
\begin{itemize}
  \item There is an $O(k^5)$-vertex kernel for the {\sc $k$-Equistability} problem that can be computed in linear time.
	      This yields an FPT algorithm for the {\sc $k$-Equistability} problem of running time $O(k^{9k+1} + m + n)$, given a graph with $n$ vertices and $m$ edges. This affirmatively answers the question posed by Levit et al.~\cite{MR3040147}.
  \item The {\sc Target-$t$ Equistability} problem admits an $O(t^2)$-vertex kernel, computable in linear time.
Moreover, there is an $O(t^{3t+1} + m + n)$ time algorithm to solve the {\sc Target-$t$ Equistability} problem.
\end{itemize}
The first result we prove in Section~\ref{sec:k-equi}, and the second in Section~\ref{sec:t-equi}.

In order to achieve the above mentioned running times of our FPT algorithms, we present a refinement of the algorithm proposed
by Levit et al.~in~\cite{MR3040147} in their proof of Theorem~\ref{thm:k-equistable}. This we present in Section~\ref{sec:XP}.

\section{Preliminaries}

\subsection{Twin classes}\label{subsec:twin-classes}

Following~\cite{MR3040147}, we say that vertices $u$ and $v$ of a graph $G$ are {\em twins} if they have exactly the same set of neighbors other
than $u$ and $v$. It is easy to verify that the twin relation is an equivalence relation.
We recall some basic properties of the twin relation~(see~\cite{MR3040147}):

\begin{lemma}
Let $G = (V, E)$ be a graph. The twin relation is an equivalence relation, and every equivalence class is either a clique or a stable set.
\end{lemma}

An equivalence class of the twin relation will be referred to as a {\em twin class}.
Twin classes that are cliques will be referred to {\em clique classes}, and the remaining classes will be referred to as {\em stable set classes}.
We say that two disjoint sets of vertices $X$ and $Y$ in a graph $G$ {\em see} each other if every vertex of $X$ is adjacent to every vertex of $Y$, and they {\em miss} each other if every vertex of $X$ is non-adjacent to every vertex of $Y$. A vertex $x$ {\em sees}
a set $Y\subseteq V(G)\setminus\{x\}$ if the singleton $\{x\}$ sees $Y$, and similarly $x$ {\em misses} $Y$ if $\{x\}$ misses~$Y$.
The set of all twin classes will be denoted by $\Pi(G)$ and referred to as the {\em twin partition} of $G$. The number of twin classes of $G$ will be denoted by $\pi(G) = |\Pi(G)|$. The following observation is an immediate consequence of the fact that the twin classes are equivalence classes under the twin relation.

\begin{observation}\label{j34fskf}
Every two distinct twin classes either see each other or miss each other.
\end{observation}

By Observation~\ref{j34fskf}, the {\em quotient graph} of $G$, denoted ${\cal Q}(G)$, is thus well defined: Its vertex set is $\Pi(G)$, and two twin classes are adjacent if and only if they see each other in $G$. Given a graph $G$, it is possible to find in linear time the twin partition $\Pi(G)$, the quotient graph ${\cal Q}(G)$ and $\pi(G)$, using any of the linear time algorithms for modular decomposition \cite{MR2500307, DBLP:conf/caap/CournierH94, MR1687819}.

The following two lemmas due to Levit et al.~\cite{MR3040147} show why twin partitions are important in the study of equistable graphs.

\begin{lemma}\label{lem:same-weight}
For every equistable function $w$ of $G$ and for every $i$, every set of the form $V^w_i= \{x\in V : w(x) = i\}$ is a subset of a twin class of $G$.
In particular, if $G$ is a $k$-equistable graph, then $\pi(G) \le k$.
\end{lemma}

\begin{corollary}\label{cor:target-t-number-of-classes}
If $G$ is a target-$t$ equistable graph, then $\pi(G) \le t$.
\end{corollary}

\begin{lemma}\label{lem:twin-same-weight}
For every equistable function $w$ of an equistable graph $G$ and for every clique class $C$ there exists an $i$ such that
$V^w_i=C$.
\end{lemma}

\subsection{Parameterized complexity}

\begin{sloppypar}
A decision problem parameterized by a problem-specific parameter~$k$ is called \emph{fixed-parameter tractable} if there exists an algorithm that solves it in time \hbox{$f(k) \cdot n^{O(1)}$}, where~$n$ is the instance size.
The function~$f$ is typically super-polynomial and depends only on~$k$.
One of the main tools to design such algorithms is the kernelization technique.
A \emph{kernelization} is a polynomial-time algorithm which transforms an instance~$(I,k)$ of a parameterized problem into an equivalent instance~$(I',k')$ of the same problem such that the size of~$I'$ is bounded by~$g(k)$ for some computable function~$g$ and $k'$ is bounded by a function of $k$.
The instance~$I'$ is said to be a \emph{kernel} of size~$g(k)$.
It is a folklore that a parameterized problem is fixed-parameter tractable if and only if it admits a kernelization. In the remainder of this paper, the kernel size is expressed in terms of the number of vertices.
For more background on parameterized complexity the reader is referred to
Downey and Fellows~\cite{DF13}.
\end{sloppypar}

\section{A refined XP-algorithm for {\sc $k$-Equistability}}\label{sec:XP}

In this section we propose a revised version of the algorithm of Levit et al.~\cite{MR3040147} for checking whether a given graph is $k$-equistable.
We implement some speed-ups and give a more careful analysis of the running time. Let us remark that this improvement does not speed up the running time when $k$ is fixed, and it is thus not relevant for the main result of Levit et al.~\cite{MR3040147}.
However, the improved running time is essential when the algorithm is applied to a kernelized instance, for the
{\sc $k$-Equistability} resp.~{\sc Target-$t$ Equistability} problem.
We refrain from formally restating the whole algorithm from~\cite{MR3040147} in order not to create redundancy.

\begin{theorem}\label{thm:refined-XP}
Let $G$ be a graph on $n$ vertices and $m$ edges, and let $k \in \mathbb N$.
Then there is an algorithm of running time
$O(n+m+\max\{n^{2k}k^{1-k},k^{3k+1}\})$
to check whether $G$ is $k$-equistable.
This algorithm computes a $k$-equistable structure, if one exists, and the same holds if a target $t$ is prescribed.
\end{theorem}

We emphasize that unlike in the statement of Theorem~\ref{thm:k-equistable}, the constant hidden in the $O$-notation
in Theorem~\ref{thm:refined-XP} does not depend on $k$ (in Theorem~\ref{thm:refined-XP}, $k$ is not restricted to be a constant).

Before we prove Theorem~\ref{thm:refined-XP}, we state the following observation.

\begin{lemma}\label{lem:technicalbound}
Let $k,n \in \mathbb N$ and let $a \in \mathbb N_0^k$ with $\sum_{i=1}^k a_i = n$.
Then $$\prod_{i=1}^k (a_i+1) \le (n/k+1)^k\,.$$
\end{lemma}

\begin{sloppypar}
\begin{proof}
If $k=1$, the statement is immediate.
So, let $k > 1$, and assume the statement is true for $k-1$.
Let $a \in \mathbb N_0^k$ with $\sum_{i=1}^k a_i = n$.
We know that $\prod_{i=1}^{k-1} (a_i+1) \le  ((n-a_k)/(k-1)+1)^{k-1}$,
and thus \hbox{$\prod_{i=1}^{k} (a_i+1) \le{}$} \hbox{$(a_k+1) \cdot ((n-a_k)/(k-1)+1)^{k-1}$}.
A straightforward calculation shows that the right hand side is maximized (over $a_k\ge 0$)
for $a_k = n/k$.
Thus,
\[
\prod_{i=1}^{k} (a_i+1) \le \left(\frac{n}{k}+1\right) \cdot \left(\frac{n-\frac{n}{k}}{k-1}+1\right)^{k-1}=\left(\frac{n}{k}+1\right)^k,
\]
which completes the proof.
\qed \end{proof}
\end    {sloppypar}

We can now prove Theorem~\ref{thm:refined-XP}.

\begin{proof}[of Theorem~\ref{thm:refined-XP}]
Recall that by Lemma~\ref{lem:same-weight}, any equistable weight function for $G$ assigns the same weight only to vertices of the same twin class. Following the algorithm of Levit et al.~\cite{MR3040147}, we proceed as follows.
First, we compute in time $O(n+m)$ the twin partition of $G$
and the quotient graph ${\cal Q}(G)$ (cf.~Section~\ref{subsec:twin-classes}).
Fix any ordering $V(G) = \{v_1,\ldots,v_n\}$ such that vertices in each twin class appear consecutively in this ordering. Clearly, the permutation of the weights within a twin class produces an equivalent weight function, i.e., a weight function is an equistable function of $G$ if and only if after any permutation of the weights within a twin class we still have an equistable function. We aim to produce a family $\mathcal{F}$ which contains all equistable functions up to permutations of the weights within a twin class. It suffices to produce all mappings $w : V(G) \rightarrow [k]$ such that the vertices in the set $w^{-1}(i)$, $i \in [k]$, appear consecutively in the ordering of $V(G)$. Let $K(n,k)$ be the number of partitions of $[n]$ into $k$ labeled intervals, where some of the intervals may be empty. It is straightforward to verify that $|\mathcal{F}|$ is bounded by $K(n,k)$.
A standard counting argument yields $K(n,k) \le k! \cdot {n+k-1 \choose n}$.

The set ${\cal F}$ can be computed in time $O(k^k(n+k-1)^{k-1})$ as follows.
Generate all one-to-one mappings from the set $[k-1]$ to an $(n+k-1)$-element set. Using the above ordering of $V(G)$,
each such mapping determines a partition of $V(G)$. If the partition refines the twin partition of $G$, then compute all the $O(k^k)$ one-to-one mappings from the resulting set of (at most $k$) non-empty intervals to the set $[k]$. Each of these mappings specifies, in a natural way, a function in ${\cal F}$.

Let us now estimate more carefully the size of ${\cal F}$. Let $\hat{n} : = \max\{n,k^2\}$. We have
\[\frac{k \cdot (\hat{n}+k)^{k-1}}{\hat{n}^k}=\frac{k}{\hat{n}} \cdot \left(1+\frac{k}{\hat{n}}\right)^{k-1} \le \frac{1}{k} \cdot \left(1+\frac{1}{k}\right)^{k-1} \le \frac{e}{k}\,,\]
implying $k \cdot (\hat{n}+k)^{k-1} \le  e \hat{n}^k/k$.
We thus obtain
\[K(n,k) \le K(\hat{n},k) \le k! \cdot {\hat{n}+k-1 \choose \hat n} = k \cdot \frac{(\hat{n}+k-1)!}{\hat{n}!} < k \cdot (\hat{n}+k)^{k-1} \le  \frac{e \hat{n}^k}{k}.\]
Thus, we have to consider only $|\mathcal{F}|=O(\hat{n}^k/k)$ many weight functions, which
can be computed in time
$O(k^k(n+k-1)^{k-1}) = O((k\hat n)^k)$.

It remains to check if any of these $O(\hat{n}^k/k)$ weight functions in $\mathcal{F}$ is an equistable function.
For every weight function $w\in {\cal F}$, the algorithm from~\cite{MR3040147} first computes the target value $t$
by evaluating the $w$-weight of an arbitrary (fixed) maximal stable set of $G$ (see~\cite{MR3040147} for details);
in our setting, this computation can be implemented in time $O(k)$.
The algorithm then computes the set $X_w$ of all $k$-dimensional vectors $x$ with integer coordinates such that
$0\le x_i\le |w^{-1}(i)|$ for all $i\in [k]$. A vector $x\in X_w$ represents the set of all subsets of $V(G)$ such that
the number of vertices of $w$-weight $i$ in the set equals $x_i$.

Note that the number of vectors in $X_w$ is bounded by $\prod_{i=1}^k (|w^{-1}(i)|+1)$, which, by Lemma~\ref{lem:technicalbound},
is in turn bounded by $(\hat{n}/k+1)^k=O((\hat{n}/k)^k)$, for each function $w$. For each vector $x\in X_w$, the algorithm then checks whether the corresponding sets are of the right weight, that is, whether $\sum_{i = 1}^k ix_i = t$ if  and only if the vector encodes a set of maximal stable sets. This latter condition can be verified in time $O(k^2)$ using the quotient graph ${\cal Q}(G)$ (see~\cite{MR3040147} for details).

\begin{sloppypar}
The running time of this algorithm is thus
\[O \left( n + m + (k\hat n)^k + \frac{\hat{n}^k}{k}\left(k+\left(\frac{\hat{n}}{k}\right)^kk^2\right)\right)\,.\]
This expression simplifies to $O(n+m+\hat{n}^{2k}{k^{1-k}}) =
O(n+m+\max\{n^{2k}k^{1-k},k^{3k+1}\})$, as desired.
\end{sloppypar}

We remark that, in case of a prescribed target value $t$, the above algorithm can be modified in an obvious way to accept only those equistable functions under which all maximal stable sets have total weight $t$. This completes the proof.
\qed \end{proof}

\section{An $O(t^2)$-vertex kernel for the {\sc Target-$t$ Equistability} problem}\label{sec:t-equi}

Given a graph $G$, the following reduction rule is specified by a positive integer $r$ as a parameter.

\begin{itemize}
  \item[] {\bf $r$-Clique Reduction}. If a clique class $C$ contains more than $r$ vertices, delete from $C$ all but $r$ vertices.
\end{itemize}
The following lemma shows why $r$-Clique Reduction rule is safe for both problems, {\sc Target-$t$ Equistability} and {\sc $k$-Equistability}.

\begin{lemma}\label{lem:clique-reduction}
Let $G$ be a graph, $T \subseteq \mathbb N$ a finite set, $C$ a clique class of $G$ with $|C|>r$ where $r := \max T$, and $k$ a positive integer.
Then, for every $t\in T$, graph $G$ is target-$t$ $k$-equistable if and only if $G'$ is target-$t$ $k$-equistable, where $G'$ is a graph obtained after the $r$-Clique Reduction rule has been applied to $G$ with respect to the clique class $C$.
\end{lemma}

\begin{proof}
Let $t\in T$. First assume that $G$ is target-$t$ $k$-equistable, say with a $k$-equistable structure $(w,t)$.
It is immediate that the restriction $w'$ of $w$ to $V(G')$ yields a $k$-equistable structure $(w',t)$ of $G'$.
Therefore $G'$ is target-$t$ $k$-equistable.

Now assume that $G'$ is target-$t$ $k$-equistable, with a $k$-equistable structure $(w',t)$.
We define a function $w : V(G) \to \{1,\ldots,k\}$ by extending $w'$ to the set $V(G)$.
Indeed, we simply put $w(u):=w'(u)$ for all $u \in V(G')$, and $w(u) := w(v)$ for all $u \in C \setminus V(G')$ where $v \in C \cap V(G')$.
The choice of $v \in C \cap V(G')$ is arbitrary, since $w'$ is constant on $C \cap V(G')$ by Lemma~\ref{lem:twin-same-weight}.

We claim that $(w,t)$ is an equistable structure of $G$.
To show this, pick an arbitrary maximal stable set $X$ of $G$.
Then $|X \cap C| \le 1$, and so we may assume that $X \subseteq V(G')$.
Clearly $X$ is a maximal stable set of $G'$, and so $w'(X) = t$.
Therefore $w(X)=t$.

Conversely, let $X \subseteq V(G)$ be a set with $w(X)=t$.
Since $w(X \cap C) \le w(X) = t$, we have $|X \cap C| \le t$.
As $w$ is constant on $C$ and $|C| > \max T \ge t$, we may w.l.o.g.~assume that $X \subseteq V(G')$.
Hence, $w'(X)=w(X)=t$, and so $X$ is a maximal stable set of $G'$.
Thus, $X$ is a maximal stable set of $G$ which completes the proof.
\qed \end{proof}

In particular, $r$-Clique Reduction rule is safe for the {\sc Target-$t$ Equistability} problem.
This is seen by putting $k:=t$ and $T=\{t\}$ in the statement of Lemma~\ref{lem:clique-reduction}.

\begin{theorem}\label{thm:t-equistable}
The {\sc Target-$t$ Equistability} problem admits a kernel of at most $t^2$ vertices, computable in linear time.
Moreover, there is an $O(t^{3t+1} + m + n)$ time algorithm to solve the {\sc Target-$t$ Equistability} problem,
given a graph with $n$ vertices and $m$ edges.
\end{theorem}
\begin{proof}
Let $G$ be a graph on $n$ vertices and $m$ edges.
Using one of the linear time algorithms for modular decomposition~\cite{MR2500307,DBLP:conf/caap/CournierH94,MR1687819}, we can compute $\Pi(G)$ and $\pi(G)$ in linear time.
If $\pi(G)>t$, then we conclude that $G$ is not target-$t$ equistable, by Corollary~\ref{cor:target-t-number-of-classes}.
Similarly, if there exists a stable set class $S$ with $|S|>t$, then we conclude that $G$ is not target-$t$ equistable.
Also, we can apply $r$-Clique Reduction rule with parameter $t$, to every clique class, in linear time.
Afterward, the graph has at most $t^2$ vertices, which proves the first statement of the theorem.

Our FPT algorithm works as follows.
First we compute in time $O(m+n)$ a kernel $G'$ with $n' \le t^2$ many vertices.
Then we apply Theorem~\ref{thm:refined-XP} to check whether $G'$ is target-$t$ equistable.
For this, we can put $k:=t$ and decide whether $G'$ is $k$-equistable with target $t$.
We thus obtain a running time of
\hbox{$O(|V(G')|+|E(G'|+\max\{n'^{2k}k^{1-k},k^{3k+1}\})$}$=O(t^{3t+1})$.
\qed \end{proof}

\section{An $O(k^5)$-vertex kernel for the {\sc $k$-Equistability} problem}\label{sec:k-equi}

This section is devoted to the proof of the following result.

\begin{theorem}\label{thm:k-equistable-fpt}
The {\sc $k$-Equistability} problem admits an $O(k^5)$-vertex kernel, computable in linear time.
Moreover, there is an $O(k^{9k+1} + m + n)$ time algorithm to solve the {\sc $k$-Equistability} problem, given a graph with $n$ vertices and $m$ edges.
\end{theorem}

\begin{proof}
Let us first prove that the second statement follows from the first one. Assume that we can compute an $O(k^5)$-vertex kernel for the {\sc $k$-Equistability} problem in linear time. By Theorem~\ref{thm:refined-XP}, we can then decide whether this kernel is $k$-equistable in time $O(k^{9k+1})$.

We now turn to the construction of the $O(k^5)$-vertex kernel.
In case of a no-instance, our algorithm simply returns a non-equistable graph, say the $4$-vertex path $P_4$.
In what follows, we will assume that the input graph $G$ satisfies $\pi(G)\le k$, since otherwise $G$ is not $k$-equistable, by Lemma~\ref{lem:same-weight}.
The following claim is the main step of our kernelization.

\begin{Claim}\label{clm:big-classes}
If there exist
two distinct twin classes $X$ and $Y$ such that
one of them is a stable set and $\min\{|X|,|Y|\}\ge k(k+1)$,
then $G$ is not $k$-equistable.
\end{Claim}

\begin{proof}
Suppose for a contradiction that $G$ is $k$-equistable, with an equistable weight function $w:V(G)\to [k]$, and that there exist
two distinct twin classes $X$ and $Y$ with $\min\{|X|,|Y|\}\ge k(k+1)$ such that $X$ is a stable set.
If the set $X\cup Y$ is contained in every maximal stable set of $G$, then $X\cup Y$ forms a twin class, a contradiction.
Thus, we may assume without loss of generality that there exists a maximal stable set $S$ of $G$
such that $X\subseteq S$ and $Y\nsubseteq S$.

Recall that $Y$ is either a clique class or a stable set class.
Since every clique intersects every stable set in at most one vertex
and every stable set class is either entirely contained in $S$ or disjoint from it,
the fact that $Y\nsubseteq S$ implies $|Y\cap S|\le 1$.
Let $i,j\in [k]$ be weights such that $|\{x\in X:w(x) = i\}|\ge k+1$, and
$|\{y\in Y:w(y) = j\}|\ge k+1$. Since $j\le k$, there exists a set $X'$ of $j$ vertices in $X$ of weight $i$.
Since $i\le k$ and $|Y\cap S|\le 1$, there exists a set $Y'$ of $i$ vertices in $Y$ of weight $j$ such that
$Y'\cap S=\emptyset$. Then, the set $S' = (S \setminus X')\cup Y'$ is not a stable set, since
otherwise by Observation~\ref{j34fskf} the set $S\cup Y$ would be a stable set properly containing $S$, contrary to the  maximality of $S$.
Note that $w(S') = w(S)$, contradicting the assumption that $w$ is an equistable weight function of $G$.
\qed \end{proof}

We consider the following two cases.\\

\noindent
\textbf{Case 1.} \emph{Every twin class $X$ with $|X|\ge k(k+1)$ is a clique class.}\\

In this case, every stable set class has less than $k(k+1)$ vertices, which implies that every maximal stable set of $G$ contains at most
$k(k+1)$ vertices from each twin class and is thus of total size at most $k^2(k+1)$.
This implies that in every $k$-equistable structure $(w,t)$ of $G$, we have $t\le k^3(k+1)$.

We now perform $r$-Clique Reduction rule from Section~\ref{sec:t-equi} with $r:=k^3(k+1)$.
By Lemma~\ref{lem:clique-reduction} applied with $T=[r]$ and $k$, the application of $r$-Clique Reduction rule is safe.
When the rule can no more be applied, we have a graph $G'$ with at most $k$ twin classes, each of size at most
$k^3(k+1)$.
We are done since $|V(G')| = O(k^5)$.\\
	
\noindent
\textbf{Case 2.} \emph{There exists a stable set twin class $X$ with $|X|\ge k(k+1)$.}\\

\noindent
By Claim~\ref{clm:big-classes}, we may assume that $X$ is the unique twin class of size at least $k(k+1)$
(since otherwise $G$ is not $k$-equistable).

Note that $V(G)\setminus X$ contains at most $k-1$ twin classes, each containing less than $k(k+1)$ vertices, 
hence $|V(G)\setminus X|\le (k-1)k(k+1) \le k^3$.

Suppose first that $X$ corresponds to an isolated vertex in the quotient graph $Q(G)$.
If $|X|<k^5$, then $|V(G)|<k^5+k^3 = O(k^5)$ and we are done.

So suppose that $|X|\ge k^5$.

\begin{Claim}\label{claim}
$G$ is $k$-equistable if and only if it admits a $k$-equistable function that is constant on $X$.
\end{Claim}

\begin{proof}
The if part being trivial, assume that $G$ is $k$-equistable, and let $(w,t)$ be a $k$-equistable structure of $G$.
Let $i \in \{1,\ldots,k\}$ be such that $|X^w_i | \ge k^4$, where $X^w_i = \{v \in X : w(v) = i\}$.
Now we define a weight function $w'$ that equals $w$ outside $X$, and is constantly $i$ on $X$.
We claim that $w'$ is a $k$-equistable function of $G$. Clearly, $w'$ is bounded by $k$.
Under $w'$, all maximal stable sets of $G$ have weight $t' := t - w(X) + w'(X)$.

The only possible problem is that $w'(S) = t'$ for some vertex set $S$ that is not a maximal stable set of $G$.
In this case, we claim that $r := |X \setminus S| \le k^4$. To see this, suppose $r > k^4$. 
Since  $|S\setminus X|\le |V(G)\setminus X|\le k^3$, we get $w'(S\setminus X)\le k^4$.
Therefore 
$w'(S) =w'(X)-w'(X\setminus S) + w'(S\setminus X)\le i(|X|-r)+ k^4$.
But $k^4 < ir$, since $i \ge 1$ and $r > k^4$. Thus
$i(|X|-r) + k^4 <  i(|X|-r) + ir = i|X| \le t'$, a contradiction.

So, $r \le k^4$, and since $k^4\le |X^w_i|$ and $w'$ is constant on $X$ we may assume that
$X \setminus S \subseteq X^w_i$. But this yields
\begin{eqnarray*}
w(S) &=& w'(S) - w'(X \cap S) + w(X \cap S)\\
&=& t' - i(|X|-r) + w(X \cap S)\\
&=& t' - i(|X|-r) + w(X \cap S) - ir + ir\\
&=& t' - i|X| + (w(X \cap S) + ir)\\
&=& t' - w'(X) + (w(X \cap S) + w(X \setminus S))\\
&=& t' - w'(X) + w(X) = t.
\end{eqnarray*}
A contradiction.
\qed \end{proof}

According to Claim~\ref{claim}, it suffices to test if $G$ is $k$-equistable by
considering all possible functions $w:V(G)\to [k]$ that are constant on $X$, and test for each of them whether it is a $k$-equistable function.

Before that, we reduce size of $X$.
For this, we compute a graph $G'$ from $G$ by deleting all but $k^4$ many vertices from $X$.
Note that, since $X$ is a twin class, $G'$ is unique up to isomorphism.

\begin{Claim}\label{clm:G-vs-G-prime}
$G$ is $k$-equistable if and only if $G'$ is $k$-equistable.
\end{Claim}

\begin{proof}
Let $X' := X \cap V(G')$ and $Y' := V(G')\setminus X'$.

First we assume that $G$ is $k$-equistable, say with an equistable structure $(w,t)$.
By Claim~\ref{claim}, we may assume that $w$ is constant on $X$, say $w|_X \equiv i$.
We now consider the weight function $w' := w|_{V(G')}$ with target value $t':=t-i|X \setminus X'|$, and claim that $(w',t')$ is a $k$-equistable structure of $G'$.
Since every maximal stable set of $G$ (resp.,~$G'$) contains $X$ (resp.,~$X'$) as a subset,
it is straightforward that every maximal stable set of $G'$ has weight $t'$.
Suppose that there is a set $S \subseteq V(G')$ with $w(S)=t'$ that is not a maximal stable set of $G'$.
Then the set $S \cup (X\setminus X')$ has total weight $t$, but is not a maximal stable set of $G$, a contradiction.
This proves that $G'$ is $k$-equistable.

Now we assume that $G'$ is $k$-equistable, say with an equistable structure $(w',t')$.
By Claim~\ref{claim} applied to $G'$, we may assume that $w'$ is constant on $X'$, say $w'|_{X'} \equiv i$.
Consider the weight function $w:V(G)\to[k]$ defined as $w(x) = w'(x)$ for all $x\in V(G')$ and $w(x) = i$ for all $x\in X\setminus X'$ with target
value $t:=t'+i|X \setminus X'|$. We claim that $(w,t)$ is a $k$-equistable structure of $G$.
Again it is straightforward that any maximal stable set of $G$ has weight $t$.
Suppose that there is a set $S \subseteq V(G)$ with $w(S)=t$ that is not a maximal stable set of $G$.

Recall that $|Y'| \le (k-1)k(k+1) \le k^3$ and consequently $w(Y')\le k^4$.
If $|X \setminus S| > k^4$, we thus obtain
\begin{eqnarray*}
  w(S) &\le& w(Y') + i|X|-i(k^4+1)  \\
  &\le& k^4 + i |X| - (k^4+1)\\
  &=& i|X|-1\\
  &<& w(X) \le t,
\end{eqnarray*}
a contradiction.
Thus, $|X \setminus S| \le k^4$, and so we may assume that $X \setminus S \subseteq X'$.
Let $S' := S \cap V(G')$.
Then $w'(S') = w(S) - i|X \setminus X'| = t'$, but $S'$ is not a maximal stable set of $G'$.
This is contradictory, and so $G$ is $k$-equistable.
\qed \end{proof}

By Claim~\ref{clm:G-vs-G-prime}, it suffices to check whether $G'$ is $k$-equistable.
Since $|V(G')|\le k^4+k^3 = O(k^4)$, we are done.

\bigskip
\begin{sloppypar}
Now, suppose that $X$ corresponds to a non-isolated vertex in the quotient graph $Q(G)$.
Then, there exists a twin class $Y$ that sees $X$.
Let $S$ be a maximal stable set of $G$ containing a vertex of $Y$.
Then, $S\subseteq V(G)\setminus X$. Since
\hbox{$|V(G)\setminus X|\le k^3$}, we have in particular
that $|S|\le k^3$.
\end{sloppypar}

If $|X|>k|S|$, then
for every $k$-equistable function $w$ of $G$ and
every maximal stable set, say $S'$, such that $X\subseteq S'$, we have
$w(S')\ge |X|>k|S|\ge w(S)$, hence $G$ is not $k$-equistable.

If $|X|\le k|S|$, then $|V(G)|\le (k+1)k^3 = O(k^4)$.

Since it is clear that the above algorithm runs in time $O(n+m)$, the proof is complete.
\qed \end{proof}


\section{Future work}

Several open problems surrounding our work remain, some of which we want to mention here in order to stimulate
research on this topic.

Firstly, we believe it is NP-hard to determine, given a graph $G$ and an integer $k$, whether $G$ is $k$-equistable.
It would be satisfying to see this proven, especially for the purpose of this paper. As mentioned in the introduction, the smallest such $k$ (if existing) might have to be exponential in the number of vertices of $G$~\cite{MR2823204}, which might serve as a hint for the hardness of this problem.

The analogous question is open also for the problem of {\sc Target-$t$ Equistability}: what is the computational complexity of
determining, given a graph $G$ and an integer $t$, whether $G$ is target-$t$ equistable? Again, the smallest such $t$ (if existing) might have to be exponential in the number of vertices of the input graph~\cite{MR2823204}.

A different computational problem in this context would be the following: given a graph $G$ and a number $k$, does it admit an equistable weight function using at most $k$ different weights? Here, both the parameterized and classical complexity are unknown.
Although we did not study this problem in depth, our impression is that it should be NP-hard, but FPT when parameterized by $k$.
In view of the results of the present paper, there might very well be a polynomial kernel for this problem.
Another problem that seems similar at first sight is whether equistability is FPT when parameterized by $\pi(G)$, the number of twin-classes of $G$.

Apart from these recognition problems, it is apparently open whether the maximum stable set problem
is FPT in the class of equistable graphs. Here we do at least know that this problem is APX-hard in this class~\cite{MR2823204}.

\bibliographystyle{abbrv}
\bibliography{bibliography}
\end{document}